\documentclass[11pt]{amsart}
\usepackage{amsmath,amssymb,amsfonts,amsbsy}
\usepackage{graphicx}
\usepackage{float}
\usepackage{setspace}
\usepackage[margin=1.0in]{geometry}

\newtheorem{thm}{Theorem}[section]

\newtheorem{cor}[thm]{Corollary}
\newtheorem{prop}[thm]{Proposition}

\theoremstyle{definition}
\newtheorem{example}[thm]{Example}
\theoremstyle{definition}
\newtheorem{defn}[thm]{Definition}
\theoremstyle{definition}
\newtheorem{remark}[thm]{Remark}

\newcommand{\mc}[1]{\mathcal{#1}}
\newcommand{\e}[1]{\emph{#1}}

\newcommand{\la}{\langle}
\newcommand{\ra}{\rangle}
\newcommand{\tr}{\mathrm{tr}}

\newcommand{\rmv}[1]{}

\newcommand{\hs}{\hskip5pt}
\newcommand{\sk}{\hskip10pt}

\newcommand{\LG}{\mc{L}(G)}
\newcommand{\RG}{\mc{R}(G)}
\newcommand{\LO}{L^1(G)}

\newcommand{\LT}{L^2(G)}

\newcommand{\LI}{L^{\infty}(G)}

\newcommand{\BH}{\mc{B}(H)}
\newcommand{\Th}{\mc{T}(H)}
\newcommand{\BLT}{\mc{B}(L^2(G))}

\newcommand{\TC}{\mc{T}(L^2(G))}

\newcommand{\NCBLT}{\mc{CB}_{\mc{L}(G)}^{\sigma,L^{\infty}(G)}(\mc{B}(L^2(G)))}
\newcommand{\NCBLTD}{\mc{CB}^{\sigma,\mc{L}(G)}_{L^{\infty}(G)}(\mc{B}(L^2(G)))}

\newcommand{\Mcb}{M_{cb}A(G)}

\newcommand{\PG}{\mc{P}(G)}

\newcommand{\PO}{\mc{P}_1(G)}

\newcommand{\vphi}{\varphi}

\newcommand{\lm}{\lambda}

\newcommand{\om}{\omega}

\newcommand{\ten}{\otimes}

\newcommand{\id}{\iota}
\newcommand{\h}[1]{\hat{#1}}
\newcommand{\Hmu}{\mc{H}_\mu}
\newcommand{\Htmu}{\tilde{\mc{H}}_\mu}
\newcommand{\Hphi}{\mc{H}_\vphi}
\newcommand{\Htphi}{\tilde{\mc{H}}_\vphi}

\providecommand{\abs}[1]{\lvert#1\rvert}
\providecommand{\norm}[1]{\lVert#1\rVert}

\newcommand{\C}{\mathbb{C}}
\newcommand{\N}{\mathbb{N}}
\newcommand{\Z}{\mathbb{Z}}
\newcommand{\R}{\mathbb{R}}

\begin{document}

\title{Quantum Channels Arising from Abstract Harmonic Analysis}
\author{Jason Crann$^{1,2}$ and Matthias Neufang$^{1,2}$}
\address{$^1$School of Mathematics \& Statistics, Carleton University, Ottawa, ON, Canada K1S 5B6}
\address{$^2$Universit\'{e} Lille 1 - Sciences et Technologies, UFR de Math\'{e}matiques, Laboratoire de Math\'{e}matiques Paul Painlev\'{e} - UMR CNRS 8524, 59655 Villeneuve d'Ascq C\'{e}dex, France}

\begin{abstract} We present a new application of harmonic analysis to quantum information by constructing intriguing classes of quantum channels stemming from specific representations of multiplier algebras over locally compact groups $G$. Beginning with a representation of the measure algebra $M(G)$, we unify and elaborate on recent counter-examples to fixed point subalgebras in infinite dimensions, as well as present an application to the noiseless subsystems method of quantum error correction. Using a representation of the completely bounded Fourier multiplier algebra $\Mcb$, we provide a new class of counter-examples to the recently solved asymptotic quantum Birkhoff conjecture, along with a systematic method of producing the examples using a geometric representation of Schur maps. Further properties of our channels including duality, quantum capacity, and entanglement preservation are discussed along with potential applications to additivity conjectures.\end{abstract}

\maketitle

\begin{spacing}{1.0}

\section{Introduction}

If $G$ is a locally compact group\let\thefootnote\relax\footnotetext{2010 \e{Mathematics Subject Classification} Primary: 43A10, 43A35, 81R05, 81R15; Secondary: 22D15, 46L89.\\[1ex] \e{Keywords}: locally compact groups, multiplier algebras; quantum channels, fixed points, minimum output entropy.\\[1ex] This work was completed as part of the Master's thesis of the first author, who was supported by an NSERC Canada Graduate Scholarship. The second author was partially supported by an NSERC Discovery Grant.}, then every measure $\mu\in M(G)$ defines a normal completely bounded map on $\BLT$ by
\begin{equation}\label{M(G)}\Theta(\mu)(x)=\int_G r_sxr_{s^{-1}}d\mu(s),\end{equation}
where $r$ is the right regular representation of $G$, and the integral converges in the weak* topology of $\BLT$. From the perspective of abstract harmonic analysis this representation has been intensely studied \cite{G,St,N}. In particular, it was shown by the second author that $\Theta$ determines a completely isometric isomorphism of $M(G)$ onto $\NCBLT$, the algebra of normal completely bounded $\LG$-bimodule maps on $\BLT$ that leave $\LI$ globally invariant \cite{N}. On the other hand, from the perspective of quantum information theory, if $\mu$ is a probability measure on $G$, then $\Theta(\mu)$ is completely positive and unital, that is, a quantum channel on $\BLT$.

In the dual setting, a representation of $\Mcb$, the completely bounded multipliers of the Fourier algebra has also been extensively studied \cite{NRS,S}. It was shown in \cite{NRS} that the resulting representation $\h{\Theta}$ defines a completely isometric isomorphism of $\Mcb$ onto $\NCBLTD$, the algebra of normal completely bounded $\LI$-bimodule maps on $\BLT$ that leave $\LG$ invariant. Moreover, if $\vphi$ is a positive definite function of norm one, $\h{\Theta}(\vphi)$ is a quantum channel on the same space $\BLT$.

The aim of this paper is to study these dual classes of quantum channels in connection with various concepts of importance in quantum information theory. We begin in section 2 with a review of the important notions from abstract harmonic analysis and infinite-dimensional quantum channels. Section 3 is dedicated to the representation of the measure algebra (\ref{M(G)}). Using results on the fixed points of the resulting channels \cite{JN}, we unify and elaborate on a class of counter-examples to a recently solved conjecture on fixed point subalgebras in infinite dimensions \cite{AGG,Lim}. We also discuss applications to the noiseless subsystems method of quantum error correction \cite{KLV}.

In section 4 we present a detailed account of the representation of $\Mcb$. With the aid of a geometric representation of Schur maps \cite{HS}, we present a systematic procedure to generate new classes of counter-examples to the recently solved asymptotic quantum Birkhoff conjecture \cite{HM}. Further properties and the duality of the channels $\Theta(\mu)$ and $\h{\Theta}(\vphi)$ are analyzed in section 5. By composing the channels we obtain a generalization of the Weyl-covariant channels studied in \cite{A,DFH}. We discuss entanglement preservation, quantum capacity, and potential applications to additivity conjectures.

\section{Preliminaries}

We assume the reader is familiar with the basics of operator algebras and operator space theory. Our references are \cite{Ren} and \cite{ER}, respectively, and we adopt their notation. In particular, we denote the set of normal--i.e., weak*-weak* continuous--completely bounded (respectively, completely positive) maps on $\BH$ by $\mc{CB}^{\sigma}(\BH)$ (respectively, $\mc{CP}^{\sigma}(\BH)$). If $M$ is a von Neumann subalgebra of $\BH$, we denote the set of normal completely bounded $M$-bimodule maps by $\mc{CB}^{\sigma}_M(\BH)$, where the latter property means that $\Phi(axb)=a\Phi(x)b$ for all $x\in\BH$, $a,b\in M$. If $\xi,\eta\in H$, we denote by $x_{\xi,\eta}$ the rank one operator $x_{\xi,\eta}(\zeta)=\la\zeta,\eta\ra\xi$, $\zeta\in H$, and by $\om_{\xi,\eta}$ the normal linear functional $\om_{\xi,\eta}(x)=\la x\xi,\eta\ra$, $x\in\BH$. The duality between the space $\Th$ of trace class operators on $H$ and $\BH$ is denoted by $\la x,\rho\ra=\tr(x\rho)$, $x\in\BH$, $\rho\in\Th$. If $M,N\subseteq\BH$, the von Neumann subalgebra of $\BH$ generated by $M$ and $N$ will be denoted $M\vee N$. The Hilbert space tensor product will be denoted $\ten_2$.

Throughout this section $G$ will denote a locally compact group with fixed left Haar measure $ds$. All Lebesgue spaces $L^p(G)$, $p\in[1,\infty]$, considered in this paper will be with respect to this fixed Haar measure, and functions $f\in\LI$ will be identified with their corresponding multiplication operator $M_f\in\BLT$. A bounded linear functional $m:\LI\rightarrow\C$ is called a \e{left invariant mean} if
\begin{equation*}\la m,1\ra=\norm{m}=1\hs\hs\text{and}\hs\hs m(l_sf)=m(f)\end{equation*}
for all $s\in G$, $f\in\LI$ where $l_sf(t)=f(s^{-1}t)$, $t\in G$. The group $G$ is said to be \e{amenable} if there is a left invariant mean on $\LI$.

A \e{unitary representation} of $G$ is a homomorphism $\pi:G\rightarrow\mc{U}(H)$, the group of unitary operators on some Hilbert space $H$, that is continuous with respect to the strong operator topology. Two important examples of unitary representations are the \e{left and right regular representations}, given respectively by $l:G\rightarrow\BLT$ and $r:G\rightarrow\BLT$, where
\begin{equation*}l_s\xi(t)=\xi(s^{-1}t)\hs\hs\text{and}\hs\hs r_s\xi(t)=\Delta(s)^{1/2}\xi(ts)\end{equation*}
for $\xi\in\LT$, $s,t\in G$, where $\Delta:G\rightarrow(0,\infty)$ is the group modular function \cite[\S2.4]{F}. It follows that
\begin{equation}\label{gvNalg}\LG=\overline{\text{span}\{l_s : s\in G\}}^{SOT}\hs\hs\text{and}\hs\hs\RG=\overline{\text{span}\{r_s : s\in G\}}^{SOT}\end{equation}
are von Neumann algebras in $\BLT$, called the \e{left and right group von Neumann algebras}, respectively. Furthermore, they satisfy the following commutation relations:
\begin{equation}\label{commute}\LG'=\RG\hs\hs\text{and}\hs\hs\RG'=\LG.\end{equation}
In relation to $\LI$, we have the following result, known as Heisenberg's theorem. We provide a simple proof for the convenience of the reader.
\begin{prop} Let $G$ be a locally compact group. Then
\begin{equation*}\LI\cap\LG=\C 1=\LI\cap\RG.\end{equation*}
Equivalently, $\LI\vee\LG=\BLT=\LI\vee\RG$.\end{prop}
\begin{proof} If $f\in\LI$ and $s\in G$, then $l_sM_f=M_{l_sf}l_s$. Thus, if $M_f\in\LI\cap\RG$ then
\begin{equation*}M_{l_sf}=l_sM_fl_{s^{-1}}=M_f\end{equation*}
for all $s\in G$, which implies that $f$ is a constant function. Hence, $\LI\cap\RG=\C 1$. Taking the commutant, we obtain
\begin{equation*}\BLT=(\C 1)'=(\LI\cap\RG)'=\LI\vee\LG.\end{equation*}\end{proof}

By lifting the left regular representation to $l:\LO\rightarrow\BLT$ via:
\begin{equation*}\la l(f),\om\ra=\int_{G}f(s)\la l_s,\om\ra ds,\sk f\in\LO,\hs\om\in\TC,\end{equation*}
it follows that
\begin{equation}\label{C^*red}\LG=\overline{\text{span}\{l(f) : f\in\LO\}}^{SOT}.\end{equation}
If instead we take the norm closure in equation (\ref{C^*red}), we obtain a $C^*$-algebra, denoted $C^*_{l}(G)$, known as the \e{left reduced $C^*$-algebra} of $G$.

A one-dimensional unitary representation $\chi:G\rightarrow\C$ is called a \e{character} of $G$. Then $\chi(G)$ is contained in the unit circle so that $\chi\in\LI$, and when $G$ is abelian, the set of characters of $G$, denoted $\h{G}$, forms a locally compact (abelian) group under multiplication when given the weak* topology of $\LI$ \cite[\S4.1]{F}, called the \e{dual} of $G$. The famous Pontrjagin duality theorem then states that $G$ is canonically isomorphic to $\h{\h{G}}$ via evaluation \cite[Theorem 4.31]{F}. One of the main tools in this theory, which allows one to transfer between $G$ and its dual, is the Fourier transform $\mc{F}:\LO\rightarrow C_0(\hat{G})$, given by
\begin{equation}\label{Fourier}\mc{F}(f)(\chi)=\int_{G}\overline{\chi(s)}f(s)ds,\sk f\in\LO,\hs\chi\in\h{G}.\end{equation}
By the Plancherel theorem, $\mc{F}$ gives rise to a unitary isomorphism $\LT\rightarrow L^2(\h{G})$ \cite[Theorem 4.25]{F}.

Given an arbitrary locally compact group, a function $\vphi:G\rightarrow\C$ is \e{positive definite} if
\begin{equation*}\sum_{i,j=1}^n\alpha_i\overline{\alpha_j}\vphi(s_j^{-1}s_i)\geq0\end{equation*}
for all $\alpha_1,..., \alpha_n\in\C$, $s_1,..., s_n\in G$ and $n\in\N$. Note that by taking $n=2$, $s_1=s$ and $s_2=e$, this condition implies that the matrix
\begin{equation*}\begin{pmatrix} \vphi(e) & \vphi(s)\\ \vphi(s^{-1}) & \vphi(e) \end{pmatrix}\end{equation*}
is positive semi-definite. Hence, $\vphi(s^{-1})=\overline{\vphi(s)}$ and $\vphi(e)^2-\vphi(s)\vphi(s^{-1})\geq 0$, so that $\abs{\vphi(s)}\leq\vphi(e)$ for all $s\in G$. In particular, positive definite functions are bounded. The set of continuous positive definite functions (respectively, of norm one) will be denoted $\PG$ (respectively, $\PO$). There is an intimate connection between elements of $\PG$ and unitary representations of $G$, namely $\vphi\in\PG$ if and only if there exists a unitary representation $\pi:G\rightarrow\BH$ and a cyclic vector $\xi\in H$ such that $\vphi(s)=\la\pi(s)\xi,\xi\ra$, $s\in G$ \cite[\S3.3]{F}. Note that when the Hilbert space dimension $d_\pi$ equals 1 (and $\norm{\xi}=1$) we recover the characters of $G$.

As usual, a \e{quantum state} on a Hilbert space $H$ is a positive trace class operator of trace one. A state on $H\ten_2 H$ is \e{separable} if it lies in the closed convex hull of product states $\rho\ten\om$, with $\rho,\om\in\Th$. In the Schr\"{o}dinger picture of dynamics, states evolve under the action of completely positive trace preserving maps $\Psi:\Th\rightarrow\Th$, while in the Heisenberg picture, states are fixed and observables evolve under normal completely positive unital maps $\Phi:\BH\rightarrow\BH$. It is known that these two pictures are equivalent, so we shall take the Heisenberg approach, making a \e{quantum channel} a normal unital completely positive map on $\BH$. In this setting, we have the following representation theorem due to Haagerup (see also \cite[Theorem 4.2]{BS}): if $\Phi\in\mc{CP}^{\sigma}(\BH)$ then there is a net $(a_i)_{i\in I}$ in $\BH$ such that $\sum_{i\in I}a_ia_i^*$ converges weak* in $\BH$ and
\begin{equation}\label{Haa}\Phi(x)=w^*-\sum_{i\in I}a_ixa_i^*,\sk x\in\BH,\end{equation}
and, conversely, every net $(a_i)_{i\in I}$ in $\BH$ such that $\sum_{i\in I}a_ia_i^*$ converges weak* in $\BH$ determines a map $\Phi\in\mc{CP}^{\sigma}(\BH)$ via equation (\ref{Haa}). For the latter statement, the convergence of (\ref{Haa}) can be seen by the Cauchy-Schwartz inequality as follows. Take $x\in\BH$ and $\om\in\Th$, $\om\geq0$. Then
\begin{equation*}\sum_{i\in I}\abs{\la a_ixa_i^*,\om\ra}\leq\sum_{i\in I}\la a_ia_i^*,\om\ra^{\frac{1}{2}}\la a_ixx^*a_i^*,\om\ra^{\frac{1}{2}}\leq\norm{x}\sum_{i\in I}\la a_ia_i^*,\om\ra<\infty.\end{equation*}
By the Jordan decomposition of trace class operators, it follows that the series (\ref{Haa}) converges weak* in $\BH$. To obtain the normality of the resulting map $\Phi$, one can use the well-known decomposition of trace class operators into a limit of rank-one operators \cite[Proposition 1.2.2]{Ren}. For details, see the proof of \cite[Lemma 2.4(i)]{BS}. The operators $(a_i)_{i\in I}$ are called the \e{Kraus operators} of $\Phi$. We say that $\Phi$ is \e{trace preserving} if $\Phi|_{\Th}:\Th\rightarrow\Th$ and $\tr(\Phi(\rho))=\tr(\rho)$ for all $\rho\in\Th$, and \e{bistochastic} if it is unital and trace preserving. In finite dimensions, our definition coincides with the usual notion of bistochasticity, which is equivalent to $\sum_{i=1}^na_ia_i^*=\sum_{i=1}^na_i^*a_i=1$. In fact, a similar relation holds in infinite dimensions, as we shall now see.

\begin{thm}\label{TP} Let $H$ be a Hilbert space, and let $\Phi\in\mc{CP}^{\sigma}(\BH)$ with Kraus operators $(a_i)_{i\in I}$. Then $\Phi$ is bistochastic if and only if $w^*-\sum_{i\in I}a_i^*a_i=w^*-\sum_{i\in I}a_ia_i^*=1$.\end{thm}
\begin{proof} As $\Phi(x)=\sum_{i\in I}a_ixa_i^*$, it suffices to show that $w^*-\sum_{i\in I}a_i^*a_i=1$ is equivalent to trace preservation. Suppose the previous sum converges to the identity. Then the map $\tilde{\Phi}:\BH\rightarrow\BH$ given by
\begin{equation*}\tilde{\Phi}(x)=\sum_{i\in I}a_i^*xa_i,\sk x\in\BH,\end{equation*}
lies in $\mc{CP}^{\sigma}(\BH)$ by Haagerup's theorem. Letting $(e_j)_{j\in J}$ be an orthonormal basis for $H$, we have $w^*-\sum_{j\in J}x_{e_j,e_j}=1$, so for $\rho\in\Th$,
\begin{align*}\tr(\Phi(\rho))&=\sum_{j\in J}\la\Phi(\rho),x_{e_j,e_j}\ra=\sum_{j\in J}\sum_{i\in I}\la a_i\rho a_i^*,x_{e_j,e_j}\ra=\sum_{j\in J}\sum_{i\in I}\la\rho,a_i^*x_{e_j,e_j}a_i\ra=\sum_{j\in J}\la\rho,\tilde{\Phi}(x_{e_j,e_j})\ra\\
                                   &=\la\rho,\tilde{\Phi}(1)\ra=\la\rho,1\ra=\tr(\rho).\end{align*}
Conversely, suppose $\Phi$ is trace preserving, and fix $\rho\in\Th$, $\rho\geq0$. Indexing by finite subsets $F$ of $I$, $a_F:=\sum_{i\in F}a_i\rho a_i^*$ defines an increasing net of positive operators satisfying $\norm{a_F}\leq\norm{\Phi(\rho)}$ for all $F$. Hence, $a_F$ converges strongly to its supremum \cite[Proposition 1.2.10]{Ren}. Since the strong operator topology is equivalent to the weak* topology on bounded convex subsets of $\BH$ \cite[Theorem 1.2.9]{Ren}, the strong operator convergence of $a_F$ implies weak* convergence. Consequently, $\sup_F a_F=\Phi(\rho)$. As $\Phi(\rho)\in\Th$, and the trace is continuous with respect to supremums of increasing positive operators, $\tr(\Phi(\rho))=\lim_F\tr(a_F)$. Thus, by trace preservation of $\Phi$,
\begin{equation*}\lim_{F}\sum_{i\in F}\la a_i^*a_i-1,\rho\ra=\lim_{F}\sum_{i\in F}\tr(a_i\rho a_i^*)-\tr(\rho)=\lim_{F}\tr(a_F)-\tr(\Phi(\rho))=0.\end{equation*}
Appealing to the Jordan decomposition of normal linear functionals on a von Neumann algebra \cite[Theorem 1.9.8]{Ren}, we obtain the result.\end{proof}

\begin{remark} The literature on infinite-dimensional quantum channels (\cite{AGG,Lim}, for instance) contains a slightly different notion of bistochasticity, namely where the operator sums $S_1=\sum_{i\in I}a_ia_i^*$ and $S_2=\sum_{i\in I}a_i^*a_i$ both converge to the identity in the strong operator topology. However, this is equivalent to our definition as the strong operator topology coincides with the weak* topology on bounded convex sets, and, hence, Theorem \ref{TP} provides the desired equivalence with trace preservation.\end{remark}

\section{Representation of $M(G)$}

The measure algebra $M(G)$ of a locally compact group $G$ is the Banach algebra of finite regular Borel measures on $G$ under the convolution product, given by
\begin{equation*}\mu\ast\nu(f)=\int_G\int_Gf(st)d\mu(s)d\nu(t),\end{equation*}
for $\mu,\nu\in M(G)$ and $f\in C_0(G)$. Recall that the representation (\ref{M(G)}) of the measure algebra is given by
\begin{equation*}\la\Theta(\mu)(x),\rho\ra=\int_G\la r_sxr_{s^{-1}},\rho\ra d\mu(s),\end{equation*}
for $\mu\in M(G)$, $x\in\BLT$, and $\rho\in\TC$. By the commutation relation (\ref{commute}), $\Theta(\mu)$ is an $\LG$-bimodule map. Moreover, it leaves $\LI$ globally invariant, and one may verify that $\Theta(\mu)(M_f)=M_{\mu\ast f}$, where
\begin{equation}\label{fdotmu}\mu\ast f(s)=\int_Gf(st)d\mu(t),\sk f\in\LI,\hs s\in G.\end{equation}
Conversely, every such map on $\BLT$ is of this form, and it follows that $\Theta$ defines a completely isometric isomorphism of $M(G)$ onto $\NCBLT$, the algebra of normal completely bounded $\LG$-bimodule maps on $\BLT$ which leave $\LI$ globally invariant.

If $G$ is a locally compact group, a bounded function $f$ on $G$ is said to be \e{left uniformly continuous} if for every $\varepsilon>0$, there exists a neighborhood $U$ of the identity such that $\norm{_sf-f}_{\infty}<\varepsilon$ for all $s\in U$, where $_sf(t)=f(st)$, $t\in G$. The set of such functions is denoted $LUC(G)$. We say that a family of functions $(f_\alpha)_{\alpha\in A}$ in $LUC(G)$ is \e{equi-left uniformly continuous} if for every $\varepsilon>0$, there exists a neighborhood $U$ of the identity such that $\norm{_sf_\alpha-f_\alpha}_{\infty}<\varepsilon$ for all $s\in U$, $\alpha\in A$. In what follows the set of probability measures in $M(G)$ will be denoted $M_1(G)$.

\begin{prop} Let $G$ be a locally compact group. If $\mu\in M_1(G)$, then $\Theta(\mu)$ is a bistochastic quantum channel on $\BLT$.\end{prop}
\begin{proof} Clearly $\Theta(\mu)$ is completely positive and unital, so it remains to show that it is trace preserving. Let $(e_i)_{i\in I}$ be an orthonormal basis for $\LT$, and let $\rho\in\TC$. Then $w^*-\sum_{i\in I}x_{e_i,e_i}=1$, and
\begin{equation*}\tr(\Theta(\mu)(\rho))=\sum_{i\in I}\la\Theta(\mu)(\rho),x_{e_i,e_i}\ra=\sum_{i\in I}\int_G\la r_s\rho r_{s^{-1}},x_{e_i,e_i}\ra d\mu(s).\end{equation*}
Let $\mc{F}(I)$ denote the directed set of finite subsets of $I$. Defining
\begin{equation*}f_F(s)=\sum_{i\in F}\la r_s\rho r_{s^{-1}},x_{e_i,e_i}\ra,\sk F\in\mc{F}(I),\hs s\in G,\end{equation*}
it follows that $(f_F)_{F\in\mc{F}(I)}$ is an equi-left uniformly continuous family in $LUC(G)$ satisfying $\sup_{F}\norm{f_F}_{\infty}<\infty$ and $f_F\rightarrow0$ pointwise, hence, $f_F\rightarrow0$ uniformly on compact sets \cite{N}. Now, using the regularity of $\mu$, one can interchange the sum and the integral to obtain
\begin{equation*}\tr(\Theta(\mu)(\rho))=\int_G\sum_{i\in I}\la r_s\rho r_{s^{-1}},x_{e_i,e_i}\ra d\mu(s)=\int_G\la r_s\rho r_{s^{-1}},1\ra d\mu(s)=\tr(\rho).\end{equation*}
Since $\rho\in\TC$ was arbitrary, the result follows.\end{proof}

To get a sense of what these channels look like, let us pursue a few concrete examples before continuing with applications.

\begin{example}\label{bf} Consider $G=\Z_2$. Any $\mu\in M_1(\Z_2)$ can be written as a convex combination of Dirac masses: $\mu=\mu(0)\delta_0+\mu(1)\delta_1$, in which case
\begin{equation*}\Theta(\mu)(\rho)=\mu(0)\rho+\mu(1)X\rho X^*,\end{equation*}
where $X=r_1={0\hs1 \choose 1\hs0}$, and $\rho\in\mc{B}(\C^2)$. Hence, we obtain the so-called bit-flip channels. Similarly, if $n\in\N$, taking $G$ to be the Cartesian product of $n$ copies of $\Z_2$, denoted $(\Z_2)^n$, yields the $n$-qubit bit flip channels. Indeed, as $r_{(s_1,...,s_n)}=\ten_{i=1}^nr_{s_i}$ for all $s=(s_1,...,s_n)\in(\Z_2)^n$, we may write
\begin{equation*}\Theta(\mu)(\rho)=\sum_{s\in(\Z_2)^n}\mu(s)X_s\rho X_s^*,\end{equation*}
where $X_s=\ten_{i=1}^nr_{s_i}$.\end{example}

\begin{example}\label{D4} Let $G=(\Z_2\times\Z_2)\rtimes_{\alpha}\Z_2$, where $\alpha$ is given by permutation. Note that this group is isomorphic to the dihedral group $D_4$. In this case $\Theta(\mu)(\rho)=\sum_{s\in G}\mu(s)r_s\rho r_{s^{-1}}$ acts on 3 qubits, and its Kraus operators are given by $\sqrt{\mu(s_1)}I\ten I\ten I$, $\sqrt{\mu(s_2)}I\ten X\ten I$, $\sqrt{\mu(s_3)}X\ten I\ten I$, $\sqrt{\mu(s_4)}X\ten X\ten I$, and
\begin{gather*}\sqrt{\mu(s_5)}\begin{pmatrix} 1 & 0 & 0 & 0\\ 0 & 0 & 1 & 0\\ 0 & 1 & 0 & 0\\ 0 & 0 & 0 & 1\end{pmatrix}\ten X,\sk\sqrt{\mu(s_6)}\begin{pmatrix} 0 & 0 & 1 & 0\\ 1 & 0 & 0 & 0\\ 0 & 0 & 0 & 1\\ 0 & 1 & 0 & 0\end{pmatrix}\ten X,\\
\sqrt{\mu(s_7)}\begin{pmatrix} 0 & 1 & 0 & 0\\ 0 & 0 & 0 & 1\\ 1 & 0 & 0 & 0\\ 0 & 0 & 1 & 0\end{pmatrix}\ten X,\sk \sqrt{\mu(s_8)}\begin{pmatrix} 0 & 0 & 0 & 1\\ 0 & 1 & 0 & 0\\ 0 & 0 & 1 & 0\\ 1 & 0 & 0 & 0\end{pmatrix}\ten X,\end{gather*}
where once again $X={0\hs1 \choose 1\hs0}$. We therefore obtain the so-called bit-flip swap channels, where the swapping of qubits is represented by the permutation action $\alpha$.\end{example}

\subsection{Fixed Point Algebras}

One reason why these channels may be of practical importance is that their fixed points have been completely characterized \cite{JN}. In fact, there is a direct connection between the fixed points of $\Theta(\mu)$ and the classical $\mu$-harmonic functions from random walk theory \cite{KV}. Let $\Hmu$ and $\Htmu$ be the fixed points of $\Theta(\mu)$ in $\LI$ and $\BLT$, respectively. That is,
\begin{equation*}\Hmu=\{f\in\LI\mid\mu\ast f=f\}\hs\hs\text{and}\hs\hs\Htmu=\{x\in\BLT\mid\Theta(\mu)(x)=x\}.\end{equation*}
Then $\Hmu$ are the $\mu$-harmonic functions in $\LI$, and we have $\Htmu\cong\Hmu\rtimes_l G$, the von Neumann crossed product of $\Hmu$ and $G$ under the action of $G$ on $\Hmu$ by left translation \cite[Proposition 6.3]{JN}. If $G$ is a $\sigma$-compact locally compact group and $\mu$ is a probability measure such that $\Hmu$ is a subalgebra of $\LI$, then $\Htmu=\Hmu\vee\LG$ \cite[\S4.4]{C}. In this case, a function $f\in\LI$ lies in $\Hmu$ if and only if it is constant on the (left) cosets of the closed subgroup generated by the support of $\mu$, denoted $G_\mu$, and we may identify $\Hmu$ with $L^{\infty}(G/G_\mu)$ \cite[Proposition 4.9]{JN}. A measure satisfying this condition is called a \e{Choquet--Deny} measure. Thus, for such measures we have
\begin{equation}\label{Htmu}\Htmu=L^{\infty}(G/G_\mu)\vee\LG.\end{equation}
It is known that every probability measure on a compact or locally compact abelian group is Choquet--Deny. We say that $\mu$ is \e{adapted} if $G_\mu=G$. In the case of Example \ref{bf}, we obtain a concrete expression for the fixed points of the $n$-qubit bit-flip channels. In particular, if $\mu$ is adapted, then $\Htmu=\mc{L}((\Z_2)^n)$.

If $H$ is a finite-dimensional Hilbert space, it is known that any bistochastic quantum channel $\Phi:\BH\rightarrow\BH$ satisfies
\begin{equation*}\text{Fix}(\Phi)=\{a_i,a_i^*\}',\end{equation*}
where $(a_i)$ are the Kraus operators of $\Phi$. In particular, the fixed point set is always a subalgebra of $\BH$ \cite[Theorem 2.1]{Kr}; see also \cite[Lemma 3.3]{BJKW}. It was shown in \cite[\S4]{AGG} that the infinite-dimensional generalization is false. However, to our knowledge, the only known counter-examples are of the form $\Theta(\mu)$ for some $\mu\in M(G)$, where $G$ is a countable discrete group (c.f. \cite[\S4]{Lim}). By the above, we know that for any $\sigma$-compact locally compact group $G$, the fixed points of $\Theta(\mu)$ form a subalgebra of $\BLT$ if and only if $\mu$ is a Choquet--Deny measure, which in the adapted case means $\Hmu=\C1$, that is, $\mu$ has trivial harmonic functions. Therefore, by \cite[Proposition 1.9]{R} (see also \cite[Theorem 4.2]{KV} and \cite[Proposition 2.1.3]{CL}), the following proposition is immediate.

\begin{prop} Let $G$ be a locally compact group. Then there exists an adapted probability measure $\mu$ such that $\Htmu$ is a subalgebra of $\BLT$ if and only if $G$ is amenable and $\sigma$-compact. In particular, $\Htmu$ is not a subalgebra of $\BLT$ for any adapted probability measure $\mu$ on a nonamenable locally compact group $G$.\end{prop}

For instance, $SL(2,\R)$ is a nonamenable $\sigma$-compact locally compact group, so that any adapted probability measure $\mu\in M_1(SL(2,\R))$ will produce a new counter-example.

Analogous to the finite-dimensional case, it was shown in \cite[Theorem 3.1]{Lim} (see also \cite[Theorem 3.5]{AGG}) that any bistochastic channel $\Phi:\BH\rightarrow\BH$ on a separable Hilbert space $H$ satisfies
\begin{equation*}\text{Fix}(\Phi)\cap\mc{K}(H)=\{a_i,a_i^*\}'\cap\mc{K}(H),\end{equation*}
where $(a_i)$ are the Kraus operators of $\Phi$, and $\mc{K}(H)$ is the ideal of compact operators. The following characterizes the above intersection for $\Phi=\Theta(\mu)$, where $\mu$ is any adapted probability measure.

\begin{prop} Let $\mu$ be an adapted probability measure on a locally compact group $G$. If $G$ is compact, then $\Htmu\cap\mc{K}(\LT)=C^*_l(G)$. If $G$ is not compact, then $\Htmu\cap\mc{K}(\LT)=\{0\}$.\end{prop}
\begin{proof} If $G$ is compact, then $\mu$ is Choquet--Deny, so $\Htmu=\LG$ by (\ref{Htmu}). But then, by \cite[Proposition 2.2]{CLR}, we have  $\Htmu\cap\mc{K}(\LT)=\LG\cap\mc{K}(\LT)=C^*_l(G)$, the right reduced $C^*$-algebra of $G$.

Suppose $G$ is not compact. If $x\in\Htmu\cap\mc{K}(\LT)$, then $x\in\tilde{\mc{H}}_{\mu^n}\cap\mc{K}(\LT)$ for all $n\in\N$ as $\Theta$ is an anti-homomorphism. Hence, for all $n\in\N$ and $\rho\in\TC$, we have
\begin{equation}\label{mu^n}\la x,\rho\ra=\la\Theta(\mu^n)(x),\rho\ra=\int_G\la r_sxr_{s^{-1}},\rho\ra d\mu^n(s).\end{equation}
Since $G$ is not compact, the convolution powers $\mu^n$ of $\mu$ converge to zero in the weak* topology of $M(G)$ \cite{M}. Also, as $x$ is compact, the function $s\mapsto\la r_sxr_{s^{-1}},\rho\ra$ lies in $C_0(G)$ for any $\rho\in\TC$, as is easily seen by approximating $x$ with finite rank operators. Thus, by weak* convergence of $\mu^n$ to 0, it follows from equation (\ref{mu^n}) that $\la x,\rho\ra=0$ for all $\rho\in\TC$, which implies $x=0$ (cf. also \cite[Proposition 2.2]{JN}).\end{proof}

Returning to finite dimensions, we finish this section with an application to the noiseless subsystems method of quantum error correction \cite{KLV}. This method relies on the well-known fact from operator algebras that any finite-dimensional unital $C^*$-algebra is unitarily equivalent to a direct sum of amplified matrix algebras, hence is of the form
\begin{equation}\label{decomp}\bigoplus_{k}M_{n_k}(\C)\ten I_{m_k},\sk k,n_k,m_k\in\N.\end{equation}
Since the fixed points of a bistochastic channel form a von Neumann subalgebra of $\BH$, where $H$ is finite-dimensional, it has such a decomposition, so the idea is to determine its explicit structure. If the decomposition contains nontrivial tensor factors, i.e., if $n_k>1$ for some $k$, then information encoded in $M_{n_k}(\C)$ will remain immune to the noise, eliminating the need for error correction.

Recall that for any probability measure $\mu$ on a locally compact group $G$, the fixed points of $\Theta(\mu)$ always contain $\LG$. Now, if $G$ is compact, the Peter-Weyl theorem implies that
\begin{equation*}\LG\cong\bigoplus_{[\pi]}M_{d_\pi}(\C)\ten I_{d_\pi},\end{equation*}
where the sum is taken over all (equivalence classes of) irreducible representations of $G$, and $d_\pi$ is the dimension of the representation \cite[Theorem 5.12]{F}. For example, if $G=D_4$, the dihedral group of order 8, we have
\begin{equation*}\mc{L}(D_4)\cong\C\oplus\C\oplus\C\oplus\C\oplus(M_2(\C)\ten I_2).\end{equation*}
Hence, $M_2(\C)$ is a noiseless subsystem for $\Theta(\mu)$ where $\mu$ is any probability measure on $D_4$. Moreover, we obtain a complete description of the noiseless subsystems for all adapted probability measures in Example \ref{D4}.

In order to implement this procedure in the laboratory, one requires the precise spatial decomposition of the Hilbert space giving rise to the decomposition (\ref{decomp}) \cite{HKL}. Fortunately, for our channels this is simply given by the coefficient functions of the irreducible representations of $G$, that is, functions of the form $\pi_{ij}(\cdot)=\la\pi(\cdot)e_j,e_i\ra$, where $(e_i)_{i=1}^{d_\pi}$ is an orthonormal basis of $H_{\pi}$. Since the irreducible representations of finite groups have been studied extensively, we arrive at a class of quantum channels whose noiseless subsystems may be computed explicitly.

\section{Representation of $\Mcb$}

If $G$ is a locally compact group, the set of coefficient functions of the left regular representation
\begin{equation*}A(G)=\{f:G\rightarrow\C\mid f(s)=\la l_s\xi,\eta\ra,\hs\xi,\eta\in\LT,\hs s\in G\},\end{equation*}
is called the \e{Fourier algebra} of $G$. It was shown by Eymard that, endowed with the norm
\begin{equation*}\norm{f}=\inf\{\norm{\xi}\norm{\eta}\mid f(\cdot)=\la\lm(\cdot)\xi,\eta\ra\},\end{equation*}
$A(G)$ is a Banach algebra under pointwise multiplication \cite[Proposition 3.4]{E}. Furthermore, it is the predual of the left group von Neumann algebra $\LG$, with $\la\lm(s),\psi\ra=\psi(s)$ for $s\in G$, $\psi\in A(G)$, so it has a natural operator space structure. A function $\vphi:G\rightarrow\C$ is a \e{completely bounded multiplier} of $A(G)$ if $m_\vphi(\psi)=\vphi\psi\in A(G)$ for all $\psi\in A(G)$ and $\norm{m_\vphi}_{cb}<\infty$. Note that these functions are automatically bounded and continuous. The set of such multipliers is denoted by $\Mcb$. See \cite{CH} for fundamental results on $\Mcb$.

Using an unpublished result of Gilbert, which implies that for every $\vphi\in\Mcb$ there exists a Hilbert space $H$ and two bounded continuous functions $\xi,\eta:G\rightarrow H$ such that
\begin{equation*}\vphi(st^{-1})=\la\eta(t),\xi(s)\ra\end{equation*}
for all $s,t\in G$, Neufang, Ruan and Spronk \cite{NRS} established the dual version of (\ref{M(G)}) by exhibiting a completely isometric algebra isomorphism
\begin{equation}\label{McbA(G)}\h{\Theta}:\Mcb\cong\NCBLTD\end{equation}
from the completely bounded Fourier multiplier algebra onto the algebra of normal completely bounded $\LI$-bimodule maps on $\BLT$ which leave $\LG$ globally invariant. If $\vphi\in\Mcb$ with Gilbert representation $\xi,\eta:G\rightarrow H$, and $(e_i)_{i\in I}$ is an orthonormal basis for $H$, the functions $\xi_i(s)=\la e_i,\xi(s)\ra$ and $\eta_i(t)=\la\eta(t),e_i\ra$ belong to $\LI$, and the map $\h{\Theta}(\vphi)$ is given by
\begin{equation*}\h{\Theta}(\vphi)(x)=w^*-\sum_{i\in I}M_{\xi_i}xM_{\eta_i},\sk x\in\BLT.\end{equation*}
Here, the invariance of $\LG$ under $\h{\Theta}(\vphi)$ follows from the relation $\h{\Theta}(\vphi)(l_s)=\vphi(s)l_s$, $s\in G$.

When $G$ is abelian with dual group $\h{G}$, the algebra $A(\h{G})$ is simply the image of the Fourier transform $\mc{F}:\LO\rightarrow C_0(\hat{G})$. In this sense, $\LO$ and $A(G)$ are dual to one another. Since it follows from Wendel's theorem that $M(G)\cong M_{cb}(\LO)$ \cite{W}, we may regard $\Mcb$ as the natural dual object of the measure algebra $M(G)$, and the representations $\Theta$ and $\h{\Theta}$ show that both algebras can be represented on the same space $\BLT$ in a manner which perfectly displays their duality.

By Bochner's theorem \cite[Theorem 4.18]{F}, the natural dual object to probability measures are positive definite functions of norm one. Given $\vphi\in\PO$, there exists a cyclic unitary representation $(\pi,\xi)$ such that $\vphi(s)=\la\pi(s)\xi,\xi\ra$, $s\in G$. Letting $(e_i)_{i\in I}$ be an orthonormal basis for $H_\pi$, and $\xi_i(s)=\la e_i,\pi(s)^*\xi\ra$, for $i\in I$, $s\in G$, it follows that
\begin{equation}\label{dualchannel}\h{\Theta}(\vphi)(x)=w^*-\sum_{i\in I}M_{\xi_i}xM_{\xi_i}^*,\sk x\in\BLT.\end{equation}
As $\h{\Theta}(\vphi)(1)=\h{\Theta}(\vphi)(l_e)=\vphi(e)1=1$, and the Kraus operators of $\h{\Theta}(\vphi)$ commute, Theorem \ref{TP} entails the following result.

\begin{prop} Let $G$ be a locally compact group. If $\vphi\in\PO$, then $\h{\Theta}(\vphi)$ is a bistochastic quantum channel on $\BLT$.\end{prop}

The fixed points of $\h{\Theta}(\vphi)$ have also been recently characterized for arbitrary locally compact groups $G$ and $\vphi\in\PO$. To this end, we let $\Hphi$ and $\Htphi$ denote the fixed points of $\h{\Theta}(\vphi)$ in $\LG$ and $\BLT$, respectively. It is known that the set $G_\vphi=\{s\in G\mid\vphi(s)=1\}$ is a closed subgroup of $G$, and $\mc{H}_\vphi=\mc{L}(G_\vphi)$ \cite[\S3.2]{CL}. Hence, by \cite[Theorem 7.5]{KNR} (see also \cite[\S4.4]{C} and \cite[Theorem 4.8]{NR})
\begin{equation}\label{vn}\tilde{\mc{H}}_\vphi=\mc{L}(G_\vphi)\vee\LI.\end{equation}
Note that contrary to the representation of $M(G)$, the fixed point set $\Htphi$ is \e{always} a von Neumann subalgebra of $\BLT$ for every $\vphi\in\PO$. In fact, it is known that the fixed points of every quantum channel on $\BH$ that has a Kraus representation with countably many commuting normal operators form a subalgebra of $\BH$ for an arbitrary Hilbert space $H$ \cite{P}.

\begin{example} Let $G=\Z_2$. Then any $\vphi\in\mc{P}_1(\Z_2)$ can be written as a convex combination of characters: $\vphi=p(0)\chi^0+p(1)\chi^1$, where $\chi^s$ is the character on $\Z_2$ corresponding to $s$, for $s=0,1$. Hence,
\begin{equation*}\h{\Theta}(\vphi)(\rho)=p(0)\rho+p(1)Z\rho Z^*,\end{equation*}
where $Z=M_{\chi^1}={1\hs0 \choose 0\hs-1}$, and $\rho\in\mc{B}(\C^2)$. We therefore obtain the so-called phase-flip channels. Similarly, if $G=(\Z_2)^n$, then any $\vphi\in\mc{P}_1((\Z_2)^n)$ can be written as a convex combination of characters \cite[\S3.3]{F}, and we obtain the $n$-qubit phase-flip channels. Indeed, if $\vphi=\sum_{s\in(\Z_2)^n}p(s)\chi^s$, we may write
\begin{equation*}\h{\Theta}(\vphi)(\rho)=\sum_{s\in(\Z_2)^n}p(s)Z_s\rho Z_s^*,\end{equation*}
where $s=(s_1,...,s_n)\in(\Z_2)^n$ and $Z_s=\ten_{i=1}^nM_{\chi^{s_i}}$, with $\chi^{s_i}$ the character on $\Z_2$ corresponding to $s_i$. Equation (\ref{vn}) therefore provides a concrete expression for the fixed points of these channels. In particular, if $\vphi$ is adapted (i.e., $G_{\vphi}=\{e\}$, see \cite[Definition 2.3]{NR}), then $\tilde{H}_\vphi=L^{\infty}((\Z_2)^n)$, meaning that $\h{\Theta}(\vphi)$ only fixes diagonal operators.\end{example}

We remark that channels of the form $\h{\Theta}(\vphi)$ for $G=\R$ have been recently considered by Holevo and Shirokov as examples of infinite-dimensional quantum channels (see \cite[Example 4]{Sh}, for instance).

\subsection{The Asymptotic Quantum Birkhoff Conjecture}

A famous theorem of Birkhoff states that the set of $d\times d$ bistochastic matrices is a convex set whose extreme points are the $d!$ permutation matrices \cite{Br}. In the quantum setting, bistochastic channels play the role of bistochastic matrices, so reasoning by pure analogy one would expect the extreme points of this convex set to be the unitary conjugations. For $d=2$ this is the case, but in general it is false \cite[Theorem 1]{LS}. Shortly after the result of \cite{LS}, it was conjectured that as the number of channel uses grows, bistochastic channels behave more like random unitaries. Formally speaking, if $\mc{RU}(M_d(\C))$ and $\mc{BIS}_d$ denote the (closed) convex hull of unitary conjugations and bistochastic channels in dimension $d$, respectively, then the asymptotic quantum Birkhoff conjecture (AQBC) states that any $\Phi\in\mc{BIS}_d$ satisfies
\begin{equation}\label{asqb}\lim_{n\rightarrow\infty}\norm{\Phi^{\ten^n}-\mc{RU}(M_d(\C)^{\ten^n})}_{cb}=0.\end{equation}
This conjecture remained open until last year, when Haagerup and Musat provided a counter-example \cite[\S3]{HM}, and, independently, Ostrev, Oza and Shor produced another one. In what follows, we will obtain a new class of counter-examples arising from the representation of $\Mcb$.

If $G$ is a finite group, then any $\vphi\in\PO$ with cyclic representation $(\pi,\xi)$ determines a $\abs{G}\times\abs{G}$ matrix $C_\vphi$ whose $(s,t)$ entry is given by $\vphi(st^{-1})=\la\pi(t)^*\xi,\pi(s)^*\xi\ra$. Since $\vphi$ is positive definite of norm one, $C_\vphi$ is a positive semi-definite matrix with diagonal entries equal to one, that is, a correlation matrix \cite{BPS}. In fact, the quantum channel $\h{\Theta}(\vphi)$ is completely determined by $C_\vphi$ in the following sense:
\begin{equation*}\h{\Theta}(\vphi)(x)=C_\vphi\circ_S x,\sk x\in\mc{B}(\LT),\end{equation*}
where $\circ_S$ denotes the Schur (or Hadamard) product of matrices. Hence, the resulting channels are Schur maps.

Let $\C_d$ denote the convex set of $d\times d$ complex correlation matrices, and let $\mc{E}(\C_d)$ denote its extreme points. The following two characterizations of $\mc{E}(\C_d)$ will be important for us.

\begin{thm}\cite[Theorem 1]{LT}\label{corr} Let $A$ be a $d\times d$ correlation matrix of rank $r$, and let $A=XX^*$ be its canonical decomposition, where $X\in M_{d\times r}(\C)$. Then $A\in\mc{E}(\C_d)$ if and only if
\begin{equation}\label{extremevphi}\mathrm{span}\{x_{\xi_j,\xi_j} : 1\leq j\leq d\}=\mc{H}_r,\end{equation}
where $\xi_j$ is the $j^{th}$ column of $X^*$, and $\mc{H}_r$ is the real linear space of $r\times r$ hermitian matrices.\end{thm}

\begin{thm}\cite[Lemma 2.4]{BPS}\label{corr2} Let $A$ be a $d\times d$ correlation matrix, and let $S_A:M_d(\C)\rightarrow M_d(\C)$ be its corresponding bistochastic Schur map. Then $A\in\mc{E}(\C_d)$ if and only if $S_A\in\mc{E}(\mc{BIS}_d)$.\end{thm}

In the case of finite groups $G$, the matrix $C_{\vphi}$ is simply the Gram matrix of the set of vectors $\{\pi(s)^*\xi : s\in G\}$, where $\vphi=(\pi,\xi)\in\PO$. Hence, in the decomposition $C_{\vphi}=XX^*$, we may take $X^*$ to be the matrix whose columns are precisely the vectors $\pi(s)^*\xi$, $s\in G$. To get a sense of when $C_{\vphi}$ yields a random unitary channel, we first consider the case of unitary maps.

\begin{prop} Let $G$ be a locally compact group, and let $\vphi=(\pi,\xi)\in\PO$. Then the quantum channel $\h{\Theta}(\vphi)$ is a unitary conjugation if and only if $\vphi$ is a character.\end{prop}
\begin{proof} If $\vphi$ is a character, then $d_\pi=1$, and by construction $\h{\Theta}(\vphi)(x)=M_{\vphi}xM_{\vphi}^*$, $x\in\BLT$, with $M_{\vphi}$ unitary. Conversely, suppose $\h{\Theta}(\vphi)=U(\cdot)U^*$ for some unitary $U\in\BLT$. As $\h{\Theta}(\vphi)$ is an $\LI$-bimodule map, there exists $\chi\in\LI$ such that $U=M_\chi$ and $\abs{\chi(s)}=1$ for all $s\in G$. But then, recalling the action of $\h{\Theta}(\vphi)$ on $\LG$, and noting $l_sM_f=M_{l_sf}l_s$ for $f\in\LI$ and $s\in G$, we have
\begin{equation*}\h{\Theta}(\vphi)(l_s)=\vphi(s)l_s=M_\chi l_sM_\chi^*=M_\chi M_{l_s\overline{\chi}}l_s\end{equation*}
for all $s\in G$. Thus, $\vphi(s)=\chi(t)\overline{\chi(s^{-1}t)}$ for all $s,t\in G$ implying $\abs{\vphi(s)}=1$ for every $s\in G$, which makes $\vphi$ a character by \cite[Theorem 32.7]{HR}.\end{proof}

In perfect duality we have the following result which follows from \cite[Theorem 3.7]{KN} together with \cite[Example 3.2]{KN} and \cite[Example 3.3]{KN}, where the corresponding result was shown for arbitrary locally compact quantum groups.

\begin{prop} Let $G$ be a locally compact group, and let $\mu\in M_1(G)$. Then the quantum channel $\Theta(\mu)$ is a unitary conjugation if and only if $\mu$ is a point mass.\end{prop}

In light of the above, one might guess that the channel $\h{\Theta}(\vphi)$ is random unitary if and only if $\vphi$ is a convex combination of characters, but as we shall see in Remark \ref{notRU}, this is not the case. We now discuss a case where $\h{\Theta}(\vphi)$ is not random unitary, which will provide a concrete counter-example to the AQBC.

\begin{example}\label{S_3} Let $G=S_3$. Then $S_3$ has a two-dimensional irreducible representation $\pi:S_3\rightarrow \mc{B}(\C^2)$, given by its action on the triangle:
\begin{gather*}
    \pi(e)=\begin{pmatrix} 1 & 0\\ 0 & 1 \end{pmatrix},\hs\hs
    \pi(123)=\begin{pmatrix} e^{i2\pi/3} & 0\\ 0 & e^{-i2\pi/3} \end{pmatrix},\hs\hs
    \pi(132)=\begin{pmatrix} e^{-i2\pi/3} & 0\\ 0 & e^{i2\pi/3} \end{pmatrix},\\\\
    \pi(12)=\begin{pmatrix} 0 & 1\\ 1 & 0 \end{pmatrix},\hs\hs
    \pi(23)=\begin{pmatrix} 0 & e^{-i2\pi/3}\\ e^{i2\pi/3} & 0 \end{pmatrix},\hs\hs
    \pi(13)=\begin{pmatrix} 0 & e^{i2\pi/3}\\ e^{-i2\pi/3} & 0 \end{pmatrix}.\end{gather*}
Let $\{e_1,e_2\}$ denote the standard basis of $\C^2$, and take $\xi=\frac{1}{\sqrt{10}}(ie_1+3e_2)$, for instance. One easily verifies that the positive definite function $\vphi(s)=\la\pi(s)\xi,\xi\ra$, $s\in G$, satisfies equation (\ref{extremevphi}). Therefore, by Theorems \ref{corr} and \ref{corr2}, the quantum channel $\h{\Theta}(\vphi)$ is an extreme point in $\mc{BIS}_6$. Furthermore, $\h{\Theta}(\vphi)$ is not random unitary as $\vphi$ is not a character on $S_3$.\end{example}

As is well-known, an element $\vphi\in\PO$ is an extreme point if and only if the associated representation $\pi$ is irreducible \cite[Thereom 3.25]{F}. By injectivity of $\h{\Theta}$, we see that this condition is necessary for $\h{\Theta}(\vphi)$ to be extreme within the set of bistochastic channels. However, it is not sufficient. Indeed, taking $\pi:S_3\rightarrow\mc{B}(\C^2)$ as in Example \ref{S_3}, and defining $\vphi=(\pi,\xi)\in\mc{P}_1(S_3)$ with $\xi=e_1$, one can easily show that equation (\ref{extremevphi}) is not satisfied. Hence, by Theorems \ref{corr} and \ref{corr2}, $\h{\Theta}(\vphi)$ is not extreme. This motivates the following definition.

\begin{defn} Let $G$ be a finite group. An element $\vphi\in\PO$ is \e{maximally extreme} if $\h{\Theta}(\vphi)$ is an extreme point of $\mc{BIS}_{\abs{G}}$.\end{defn}

For example, the positive definite function $\vphi$ from Example \ref{S_3} satisfies equation (\ref{extremevphi}), and is therefore maximally extreme. Being dual to the random unitary channels $\Theta(\mu)$ arising from probability measures in $M(S_3)$, one would expect that the resulting channel $\h{\Theta}(\vphi)$ is ``far away'' from the set of all random unitaries in six dimensions. This is indeed the case, and with the help of factorizable maps \cite{HM}, it is not difficult to prove.

In \cite[Corollary 2.3]{HM}, Haagerup and Musat showed that every $\Phi\in\mc{BIS}_d$ which is extreme within the set of completely positive maps in $M_d(\C)$ is not factorizable. Since the channels $\h{\Theta}$ have commutating Kraus operators, being extreme in $\mc{CP}(M_d(\C))$ is the same as being extreme in $\mc{BIS}_d$ \cite[Corollary 1]{LS}. Thus, if $\vphi\in\PO$ is maximally extreme with a cyclic representation $(\pi,\xi)$ satisfying $d_\pi\geq 2$, then $\h{\Theta}(\vphi)$ is not factorizable. As shown in \cite[Theorem 6.1]{HM}, tensoring a bistochastic channel increases its $cb$-norm distance to the set of factorizable maps on $M_d(\C)$. Hence, we immediately obtain the following.

\begin{thm}\label{notF} Let $G$ be a finite group. If $\vphi\in\PO$ is maximally extreme with a cyclic representation $(\pi,\xi)$ satisfying $d_\pi\geq 2$, then $\h{\Theta}(\vphi)$ does not satisfy the asymptotic quantum Birkhoff conjecture.
\end{thm}

\begin{remark} It was shown in \cite{Ri} that real correlation matrices give rise to factorizable Schur maps. Thus, real-valued positive definite functions are not maximally extreme, and one must consider complex-valued functions in order to apply Theorem \ref{notF}.\end{remark}

By construction it follows that the positive definite function from Example \ref{S_3} is maximally extreme, thereby giving a concrete counter-example to the AQBC. The associated correlation matrix is given by

\begin{equation*}
C_{\vphi}=\begin{pmatrix} 1 & -\frac{1}{2}+\frac{2}{5}\sqrt{3}i & -\frac{1}{2}-\frac{2}{5}\sqrt{3}i & 0 & -\frac{3}{10}\sqrt{3} & \frac{3}{10}\sqrt{3}\\
-\frac{1}{2}-\frac{2}{5}\sqrt{3}i & 1 & -\frac{1}{2}+\frac{2}{5}\sqrt{3}i & \frac{3}{10}\sqrt{3} & 0 & -\frac{3}{10}\sqrt{3}\\
-\frac{1}{2}+\frac{2}{5}\sqrt{3}i & -\frac{1}{2}-\frac{2}{5}\sqrt{3}i & 1 & -\frac{3}{10}\sqrt{3} & \frac{3}{10}\sqrt{3} & 0\\
0 & \frac{3}{10}\sqrt{3} & -\frac{3}{10}\sqrt{3} & 1 & -\frac{1}{2}-\frac{2}{5}\sqrt{3}i & -\frac{1}{2}+\frac{2}{5}\sqrt{3}i\\
-\frac{3}{10}\sqrt{3} & 0 & \frac{3}{10}\sqrt{3} & -\frac{1}{2}+\frac{2}{5}\sqrt{3}i & 1 & -\frac{1}{2}-\frac{2}{5}\sqrt{3}i\\
\frac{3}{10}\sqrt{3} & -\frac{3}{10}\sqrt{3} & 0 & -\frac{1}{2}-\frac{2}{5}\sqrt{3}i & -\frac{1}{2}+\frac{2}{5}\sqrt{3}i & 1\\\end{pmatrix}.\end{equation*}

We note that, although relying on important results from \cite{HM}, our counter-examples are different from the ones considered in that paper. To see this, recall that the correlation matrix $C_{\vphi}$ has entries $\vphi(st^{-1})$, $s,t\in G$. Thus, each column contains the values of the function $\vphi$, permuted according to the group multiplication. It is easily checked that the counter-examples considered in \cite{HM} are not of this form. Moreover, we can produce counter-examples in a systematic fashion by using a geometric characterization of maximal extremity. As an illumination for the counter-example above, see Figure 1.

\subsection{Bloch Sphere Representation of Schur Maps}

Let $A$ be a $d\times d$ correlation matrix of rank $r$, and let $A=XX^*$ be its canonical decomposition, where $X\in M_{d\times r}(\C)$. Denoting by $\xi_j$ the $j^{th}$ column of $X^*$, we may study the associated Schur map $S_A$ on $M_d(\C)$ by looking at the set of rank 1 projections $\{x_{\xi_j,\xi_j} : 1\leq j\leq d\}$ on the generalized Bloch sphere \cite{Ki}. Specifically, let $\sigma:=(\sigma_1,...,\sigma_{r^2-1})$ be a vector of generators for $SU(r)$. Then any state $\rho\in M_r(\C)$ can be uniquely characterized by
\begin{equation*}\rho=\frac{1}{r}I_r+\frac{1}{2}v_\rho\cdot\sigma\end{equation*}
with $v_\rho\in\R^{r^2-1}$ satisfying $\norm{v_\rho}\leq\sqrt{\frac{2(r-1)}{r}}$, and $v_\rho\cdot\sigma$ denoting the pointwise dot product of $v_\rho$ and $\sigma$. The correlation matrix $A$ then yields $d$ vectors $v_j$ on the Bloch Sphere corresponding to the pure states $x_{\xi_j,\xi_j}$, and we gain insight into the bistochastic channel $S_A$ by studying the geometry of the resulting set. One result of particular interest is the following. It may be proved by combining Theorem \ref{corr} with the above construction, but see \cite[\S II]{HS} for a detailed argument.

\begin{thm}\label{BSR} Let $A$ be a $d\times d$ correlation matrix of rank $r$. Then $A\in\mc{E}(\C_d)$ if and only if the affine span of $\{v_j\mid1\leq j\leq d\}$ is $\R^{r^2-1}$.\end{thm}

Applying the above to our case, if $G$ is a finite group and $\vphi=(\pi,\xi)$ is a positive definite function of norm one, we obtain $\abs{G}$ many vectors on the Bloch sphere in $\R^{d_\pi^2-1}$ corresponding to the pure states $x_{\pi(s)^*\xi,\pi(s)^*\xi}$, $s\in G$. If the affine span of the resulting vectors is all of $\R^{d_\pi^2-1}$, then we obtain a counter-example to the AQBC. For a visual example, consider the two-dimensional irreducible representation $\pi:S_3\rightarrow\mc{B}(\C^2)$ from Example \ref{S_3}. As $\xi\in\C^2$ varies, we see that the orbit of $x_{\xi,\xi}$ varies accordingly in its geometric structure:
\begin{figure}[h]
\centering
\includegraphics[scale=0.20]{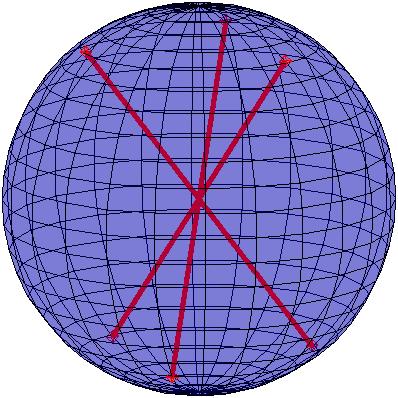}
\vskip10pt
\caption{$\vphi=(\pi,\xi)$ with $\xi=\frac{1}{\sqrt{10}}(i,3)$ is maximally extreme. Produced in Maple.}
\end{figure}
\begin{figure}[h]
\centering
\includegraphics[scale=0.20]{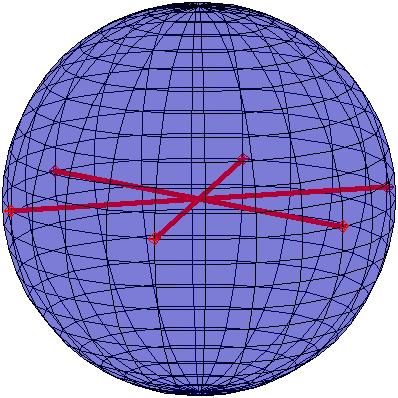}
\vskip10pt
\caption{$\vphi=(\pi,\xi)$ with $\xi=\frac{1}{\sqrt{2}}(1,i)$ is not maximally extreme. Produced in Maple.}
\end{figure}

\begin{remark} Given a finite group $G$ and $\vphi=(\pi,\xi)\in\PO$, it would be interesting to study the relationship between the distance of the Schur map $\h{\Theta}(\vphi)$ to the set of random unitaries and the volume of the corresponding orbit on the Bloch sphere. This has been done for a class of Schur maps in \cite{HS}, showing remarkable similarities.\end{remark}

In \cite[Theorem 4]{LS}, it was shown that if $\Phi\in\mc{BIS}_d$ has a representation with two linearly independent Kraus operators, then it is unitarily equivalent to a Schur map, i.e., there exists a unitary $U\in M_d(\C)$ such that $\Psi(x)=U^*\Phi(UxU^*)U$, defines a Schur map on $M_d(\C)$. Consequently, the Bloch sphere representation yields the following dichotomy.

\begin{thm}\label{bloch} If $\Phi\in\mc{BIS}_d$ has a representation with two linearly independent Kraus operators, then either $\Phi$ is extreme in $\mc{BIS}_d$, or it is random unitary.\end{thm}
\begin{proof}Let $\Psi$ be the bistochastic Schur map obtained from $\Phi$ via unitary equivalence. As $\Psi$ is also represented with two linearly independent Kraus operators, the rank of the associated correlation matrix $A=XX^*$ is two, and we obtain $d$ vectors $v_j$ on the Bloch sphere in $\R^3$ representing the pure states $x_{\xi_j,\xi_j}$, where $\xi_j$ is the $j^{th}$ column of $X^*$. If the affine span of the $v_j$'s is all of $\R^3$, then $\Psi$ is extreme by Theorem \ref{BSR} and Theorem \ref{corr2}. If not, then all the vectors $v_j$ lie on the same plane and there exists rotations $R_j\in SU(2)$ satisfying $R_jx_{\xi_1,\xi_1}R_j^*=x_{\xi_j,\xi_j}$ and whose action on the Bloch sphere takes $v_1$ to $v_j$ for all $1\leq j\leq d$. Indeed, if $\theta_j$ is the angle between $v_1$ and $v_j$, and $n=(n_x,n_y,n_z)$ is the unit normal to the plane, then
\begin{equation*}R_j=\cos\bigg(\frac{\theta_j}{2}\bigg)I_2-i\sin\bigg(\frac{\theta_j}{2}\bigg)\bigg(n_xX+n_yY+n_zZ\bigg),\end{equation*}
where $X,Y,Z$ are the $2\times 2$ Pauli matrices \cite[\S8.3]{NC}. As these rotations commute, there is a common orthonormal eigenbasis $\{e_1,e_2\}$ of $\C^2$. Let $\lm^1_j$ and $\lm^2_j$ denote the eigenvalues of $R_j$ corresponding to $e_1$ and $e_2$, respectively, for $1\leq j\leq d$, and define $M_1=\mathrm{diag}(\alpha_1\overline{\lm^1_1},...,\alpha_d\overline{\lm^1_d})$ and $M_2=\mathrm{diag}(\alpha_1\overline{\lm^2_1},...,\alpha_d\overline{\lm^2_d})$, where $\alpha_j$ are those elements on the unit circle such that $R_j\xi_1=\alpha_j\xi_j$, $1\leq j\leq d$. Clearly, $M_1$ and $M_2$ are both unitaries, and if $\{b_1,...,b_d\}$ is the standard basis for $\C^d$ we have
\begin{align*}\la\Psi(\rho)b_j,b_i\ra&=\la\xi_j,\xi_i\ra\la\rho b_j,b_i\ra\\
                                     &=\bigg(\la\xi_j,e_1\ra\la e_1,\xi_i\ra+\la\xi_j,e_2\ra\la e_2,\xi_i\ra\bigg)\la\rho b_j,b_i\ra\\
                                     &=\bigg(\overline{\alpha_j}\alpha_i\la\alpha_j\xi_j,e_1\ra\la e_1,\alpha_i\xi_i\ra+\overline{\alpha_j}\alpha_i\la\alpha_j\xi_j,e_2\ra\la e_2,\alpha_i\xi_i\ra\bigg)\la\rho b_j,b_i\ra\\
                                     &=\bigg(\overline{\alpha_j}\alpha_i\la\xi_1,R_j^*e_1\ra\la R_i^*e_1,\xi_1\ra+\overline{\alpha_j}\alpha_i\la\xi_1,R_j^*e_2\ra\la R_i^*e_2,\xi_1\ra\bigg)\la\rho b_j,b_i\ra\end{align*}\newpage
                                     \begin{align*}
                                     &=\bigg(\abs{\la\xi_1,e_1\ra}^2\overline{\alpha_j}\lm^1_j\alpha_i\overline{\lm^1_i}+\abs{\la\xi_1,e_2\ra}^2\overline{\alpha_j}\lm^2_j\alpha_i\overline{\lm^2_i}\bigg)\la\rho b_j,b_i\ra\\
                                     &=\abs{\la\xi_1,e_1\ra}^2\la M_1\rho M_1^*b_j,b_i\ra+\abs{\la\xi_1,e_2\ra}^2\la M_2\rho M_2^*b_j,b_i\ra.\end{align*}
Therefore, $\Psi(\rho)=\abs{\la\xi_1,e_1\ra}^2M_1\rho M_1^*+\abs{\la\xi_1,e_2\ra}^2M_2\rho M_2^*$ for all $\rho\in M_d(\C)$, and $\Psi$ is random unitary. Since the inclusions $\Psi\in\mc{RU}(M_d(\C))$ and $\Psi\in\mc{E}(\mc{BIS}_d)$ are preserved under unitary equivalence, the same results hold true for our original channel $\Phi$.\end{proof}

\begin{remark}\label{notRU} Applying this result to the positive definite function $\vphi$ in Figure 2, we see that the resulting channel is random unitary. However, coming from an irreducible representation of $S_3$ of dimension $2$, $\vphi$ cannot lie in the convex hull of characters. The characterization of positive definite functions $\vphi$ for which $\h{\Theta}(\vphi)$ is random unitary is therefore more subtle then at first glance.\end{remark}

As a final application of the geometric representation, we will provide a simplified expression for the quantum capacity of Schur maps. Although this has been studied quite extensively \cite{DS,WP,DBF,HHLZ}, especially in the memory setting, the connection with the geometric representation may provide insight into capacities of other quantum channels. We present the result for our class of channels but it equally holds for all Schur maps by the same argument.

Given a finite-dimensional completely positive trace preserving (CPTP) map $\Phi:M_d(\C)\rightarrow M_d(\C)$ of the form $\Phi(x)=\sum_{i=1}^na_ixa_i^*$, $x\in M_d(\C)$, the \e{complement} of $\Phi$ is the CPTP map $\Phi^c:M_d(\C)\rightarrow M_n(\C)$ given by $\Phi^c(x)=\sum_{j=1}^db_jxb_j^*$, where $x\in M_d(\C)$ and $[b_j]_{ik}=[a_i]_{jk}$ \cite{KMNR}. In this more general context, the operators $(a_i)_{i=1}^n$ and $(b_j)_{j=1}^d$ are still called Kraus operators. The \e{quantum capacity} of $\Phi$ is then
\begin{equation*}Q(\Phi)=\lim_{n\rightarrow\infty}\frac{1}{n}\sup_{\rho}J(\Phi^{\ten n},\rho),\end{equation*}
where $J(\Phi,\rho)=S(\Phi(\rho))-S(\Phi^c(\rho))$ is the \e{coherent information}, and $S(\rho)=-\tr(\rho\log\rho)$ is the von Neumann entropy of the state $\rho$ \cite{De,L,Shor2}. Physically, $Q(\Phi)$ may be viewed as the amount of quantum information that can be reliably transmitted by $\Phi$, asymptotically.

\begin{prop} Let $G$ be a finite group and $\vphi=(\pi,\xi)\in\PO$. Then
\begin{equation}\label{Q}Q(\h{\Theta}(\vphi))=\max_{\mu\in M_1(G)}\bigg(H(\mu)-S\bigg(\sum_{s\in G}\mu(s)x_{\pi(s)\xi,\pi(s)\xi}\bigg)\bigg),\end{equation}
where $H(\mu)=-\sum_{s\in G}\mu(s)\log(\mu(s))$ is the classical Shannon entropy of $\mu$.\end{prop}
\begin{proof} The complement of $\h{\Theta}(\vphi)$ is the map $\h{\Theta}(\vphi)^c:\BLT\rightarrow M_{d_\pi}(\C)$ with Kraus operators $(N_s)_{s\in G}\subseteq M_{d_\pi,|G|}(\C)$ satisfying $[N_s]_{it}=[M_{\xi_i}]_{st}=\delta_{s,t}\xi_i(s)=\delta_{s,t}\la e_i,\pi(s)\xi\ra$, where $(e_i)_{i=1}^{d_\pi}$ is the standard orthonormal basis of $H_\pi$, and $\xi_i$ are as in equation (\ref{dualchannel}). Thus, denoting the standard basis of $\LT$ by $(\delta_s)_{s\in G}$, for every $\rho\in\BLT$ we have
\begin{align*}\la\h{\Theta}(\vphi)^c(\rho)e_j,e_i\ra&=\sum_{s\in G}\la\rho N_s^*e_j,N_s^*e_i\ra=\sum_{s\in G}\xi_i(s)\overline{\xi_j(s)}\la\rho\delta_s,\delta_s\ra\\
                                                    &=\sum_{s\in G}\la e_i,\pi(s)^*\xi\ra\la\pi(s)^*\xi,e_j\ra\la\rho\delta_s,\delta_s\ra=\sum_{s\in G}\la x_{\pi(s)^*\xi,\pi(s)^*\xi}^te_j,e_i\ra\la\rho\delta_s,\delta_s\ra,\end{align*}
where $t$ denotes the transpose. Therefore, $\h{\Theta}(\vphi)^c(\rho)=\sum_{s\in G}\la\rho\delta_s,\delta_s\ra x_{\pi(s)^*\xi,\pi(s)^*\xi}^t$. In particular we see that if $M_\chi\in\LI$ with $|\chi(s)|=1$ for all $s\in G$, then $\h{\Theta}(\vphi)^c(M_{\chi}^*\rho M_{\chi})=\h{\Theta}(\vphi)^c(\rho)$ for every $\rho\in\BLT$.

Identifying $\LT$ with $L^2(\Z_{|G|})$, one can easily see that the conditional expectation from $\BLT$ onto $\LI$, denoted $E$, is a random unitary channel with Kraus operators $\frac{1}{\sqrt{|G|}}M_{\chi^t}$, $t\in\Z_{|G|}$. Since Schur maps are degradable, and the coherent information is concave on degradable channels \cite[Appendix B]{DS}, we may apply an argument from \cite{WP} to obtain
\begin{align*}J(\h{\Theta}(\vphi),E(\rho))&\geq\frac{1}{|G|}\sum_{t\in\Z_{|G|}}J(\h{\Theta}(\vphi),M_{\chi^t}^*\rho M_{\chi^t})\\
                                                             &=\frac{1}{|G|}\sum_{t\in\Z_{|G|}}\bigg(S(\h{\Theta}(\vphi)(M_{\chi^t}^*\rho M_{\chi^t}))-S(\h{\Theta}(\vphi)^c(\rho))\bigg)\\
                                                             &=\frac{1}{|G|}\sum_{t\in\Z_{|G|}}\bigg(S(M_{\chi^t}^*\h{\Theta}(\vphi)(\rho )M_{\chi^t})-S(\h{\Theta}(\vphi)^c(\rho))\bigg)\\
                                                             &=J(\h{\Theta}(\vphi),\rho).\end{align*}
Hence, as observed in \cite{WP} for two dimensional channels, the coherent information is maximized on diagonal input states. As these are of the form $M_{\mu}$ for some probability measure $\mu$ on $G$, and the transpose is entropy preserving, we obtain the desired formula by using the fact that $Q(\Phi)=\sup_{\rho}J(\Phi,\rho)$ for degradable channels \cite[Appendix B]{DS}.\end{proof}

\section{Duality and Further Properties}

A distinguishing property of quantum channels coming from our representations is that they have ``dual'' counterparts acting on the same space. One manifestation of this duality is the following result, where for a set $\mc{S}\subseteq\mc{CB}(\BLT)$, we denote by $\mc{S}^c$ its commutant in $\mc{CB}(\BLT)$.

\begin{thm}\cite[Theorem 5.2]{NRS}\label{comm} Let $G$ be a locally compact group. Then we have
\begin{equation*}\h{\Theta}(\Mcb)=\Theta(M(G))^c\cap\mc{CB}^{\sigma}_{\LI}(\BLT).\end{equation*}\end{thm}

\begin{cor}\label{commp} Let $G$ be a locally compact group. Then we have
\begin{equation*}\h{\Theta}(\PO)=\Theta(M_1(G))^c\cap\mc{UCP}^{\sigma}_{\LI}(\BLT).\end{equation*}\end{cor}
\begin{proof} The inclusion ``$\subseteq$'' follows directly from Theorem~\ref{comm}. Conversely, as $M(G)$ is the dual of the $C^*$-algebra $C_0(G)$, it is linearly spanned by its states, i.e., the probability measures. Thus, if a map $\Phi\in\mc{UCP}^{\sigma}_{\LI}(\BLT)$ commutes with $\Theta(M_1(G))$, it necessarily commutes with $\Theta(M(G))$, ensuring it lies in $\h{\Theta}(\Mcb)\cap\mc{UCP}^{\sigma}_{\LI}(\BLT)$, which, by \cite[Lemma 4.1]{NRS} implies that $\Phi\in\h{\Theta}(\PO)$.\end{proof}

From the viewpoint of quantum information, Corollary~\ref{commp} says that the dual algebras $M(G)$ and $\Mcb$ give rise to commuting quantum channels. Moreover, any $\LI$-bimodule channel commuting with channels arising from $M(G)$, is \e{necessarily} a channel arising from $\Mcb$.

When $G$ is a locally compact abelian group, this duality can be expressed directly via the Fourier transform. Indeed, if $\mc{F}:\LT\rightarrow L^2(\h{G})$ denotes the unitary isomorphism from Plancherel's theorem (see equation (\ref{Fourier})), it follows that $\mc{F}r_s=M_{\chi^s}\mc{F}$ for all $s\in G$. Moreover, the Fourier transform can be extended to $M(G)$ in such a way that $\mu\in M_1(G)$ if and only if its image $\h{\mu}\in\mc{P}_1(\h{G})$ \cite[\S4.1]{F}. Using properties of $\h{\Theta}$ one can then easily show that $\h{\Theta}(\h{\mu})(\h{x})=\mc{F}(\Theta(\mu)(\mc{F}^*\h{x}\mc{F}))\mc{F}^*$ for  $\h{x}\in\mc{B}(L^2(\h{G}))$, thus providing a unitary equivalence between $\h{\Theta}(\h{\mu})$ and $\Theta(\mu)$ (see also \cite[Theorem 5.2]{S2}). As the quantum capacity is preserved under unitary equivalence, one may therefore use equation (\ref{Q}) to evaluate the quantum capacity of $\Theta(\mu)$ for any $\mu\in M_1(G)$, when $G$ is finite. Most notably, when $G=(\Z_2)^n$ for some $n\in\N$, this provides a simple maximization to determine the quantum capacity of an arbitrary $n$-qubit bit-flip channel.

In general, any quantum channel $\Phi:M_d(\C)\rightarrow M_d(\C)$ may be viewed as a map in $\mc{CB}(\BLT)$ for a finite abelian group $G$ of size $d$. One may then define its ``Fourier transform'' $\h{\Phi}$ by the above procedure in order to gain new insight. This may be particularly useful in the case of $n$-qubit channels, where the canonical group associated to the Hilbert space is $(\Z_2)^n$.

One of the biggest open problems in quantum information theory is to provide a concrete counter-example to the minimum output entropy (MOE) conjecture, which states that
\begin{equation*}S_{\min}(\Phi\ten\Psi)=S_{\min}(\Phi)+S_{\min}(\Psi),\end{equation*}
for all (finite-dimensional) CPTP maps $\Phi$ and $\Psi$, where $S_{\min}(\Phi)$ is the infimum of $S(\Phi(\rho))$ over all states $\rho$. Using random matrix techniques, it was shown by Hastings that in general this is false \cite{H}, but the search for a concrete counter-example is still underway.

Suppose that $G$ is a finite group. Then the rank 1 projection onto the constant functions in $\LT$ lies in $\LG$. Indeed, one may verify that $\frac{1}{|G|}\sum_{s\in G}l_s$ is the required projection. As $\LG\subseteq\Htmu$ for all $\mu\in M_1(G)$, it follows that $\Theta(\mu)$ always fixes a pure state. Hence, $S_{\min}(\Theta(\mu))=\inf_{\rho}S(\Theta(\mu)(\rho))=0$. On the other hand, $\LI$ contains many rank 1 projections, and since $\LI\subseteq\tilde{\mc{H}}_\vphi$ for every $\vphi\in\PO$, the channel $\h{\Theta}(\vphi)$ always fixes a pure state. Consequently, $S_{\min}(\h{\Theta}(\vphi))=0$. Therefore, the channels $\Theta(\mu)$ and $\h{\Theta}(\vphi)$ are of little interest to the conjecture as zero MOE implies additivity with all other channels. However, the composition $\Theta(\mu)\circ\h{\Theta}(\vphi)$ often has positive MOE. For instance, taking $G=\Z_d$, we obtain a subclass Weyl-covariant channels \cite{A,DFH}, many of which have positive MOE. Indeed,
\begin{equation*}\Theta(\mu)\circ\h{\Theta}(\vphi)(\rho)=\sum_{s,t\in\Z_d}\mu(s)p(t)r_sM_{\chi^t}\rho M_{\chi^t}^*r_{s^{-1}},\end{equation*}
for $\rho\in\mc{B}(\C^d)$, where $\sum_{t\in\Z_d}p(t)\chi^t$ is the convex decomposition of $\vphi$. It is known that Weyl-covariant channels are of the form
\begin{equation}\label{WCC}\Phi(\rho)=\sum_{s,t\in\Z_n}q(s,t)r_sM_{\chi^t}\rho M_{\chi^t}^*r_{s^{-1}},\end{equation}
for an arbitrary probability measure $q\in M_1(\Z_d\times\Z_d)$ \cite[\S4]{DFH}, so we see that our composed channels are precisely the ``separable'' Weyl-covariant channels.

It would be interesting to study a generalization of Weyl-covariant channels by considering compositions of the form $\Theta(\mu)\circ\h{\Theta}(\vphi)$ for nonabelian groups $G$. One result in this direction, which is of interest in its own right, concerns the entanglement breaking properties of our channels. Recall that a finite-dimensional CPTP map $\Phi$ is \e{entanglement breaking} if $(\id\ten\Phi)(\rho)$ is separable for arbitrary states $\rho$ \cite{HSR}. These channels have many important properties, and they are one of the only classes of channels for which additivity is known to hold \cite{Shor}.

\begin{thm}\label{EB} Let $G$ be a finite group, and let $\vphi\in\PO$. Then $\h{\Theta}(\vphi)$ is entanglement breaking if and only if $\vphi=\delta_e$.\end{thm}
\begin{proof} If $\vphi=\delta_e$, then $\h{\Theta}(\vphi)$ is the conditional expectation onto $\LI$ which is known to be entanglement breaking. On the other hand, if $\h{\Theta}(\vphi)$ is entanglement breaking then $\rho:=(\id\ten T\circ\h{\Theta}(\vphi))(x_{\eta,\eta})$ is a positive operator by \cite[Theorem 4]{HSR}, where $T$ is the transpose map and $\eta=\frac{1}{\sqrt{|G|}}\sum_{s\in G}\delta_s\ten\delta_s$ is the maximally entangled vector in $L^2(G\times G)$. Hence,
\begin{equation*}0\leq\rho=\sum_{s,t\in G}\vphi(st^{-1})x_{\delta_s,\delta_t}\ten x_{\delta_t,\delta_s}.\end{equation*}
Fix $s\in G$, $s\neq e$, and let $\alpha\in\C$ be non-zero. Choosing $\xi=\delta_e\ten\delta_s+\alpha\delta_s\ten\delta_e\in L^2(G\times G)$, it follows that $\la\rho\xi,\xi\ra=\overline{\alpha}\vphi(s)+\alpha\vphi(s^{-1})=\overline{\alpha}\vphi(s)+\alpha\overline{\vphi(s)}=2\Re(\overline{\alpha}\vphi(s))\geq0$. Since $\alpha$ was arbitrary, we must have $\vphi(s)=0$. Since $\vphi(e)=1$, the result follows.\end{proof}\newpage

\begin{thm}\label{EB2} Let $G$ be a finite group, and let $h$ be the Haar measure of $G$. If $\Theta(\mu)$ is entanglement breaking for some $\mu\in M_1(G)$, then $G$ is abelian, and $\mu=h$.\end{thm}
\begin{proof} Suppose $\Theta(\mu)$ is entanglement breaking. It follows from \cite[Theorem 6]{HSR} that the support of $\mu$ is all of $G$. In this case the convolution powers $\mu^n$, $n\geq 1$, converge to $h$ in the norm topology of $M(G)$, so that $\Theta(\mu^n)$ converges to $\Theta(h)$ in $cb$-norm. Since the convex set of separable states is norm closed in $\BLT$, it follows that $(\id\ten\Theta(h))(\rho)=\lim_n(\id\ten\Theta(\mu^n))(\rho)$ is separable for any $\rho$. Thus, $\Theta(h)$ is entanglement breaking.

Let $\eta=\frac{1}{\sqrt{|G|}}\sum_{s\in G}\delta_s\ten\delta_s$ be the maximally entangled vector in $L^2(G\times G)$. Then by \cite[Theorem 4]{HSR} we know that $(\id\ten T\circ\Theta(h))(x_{\eta,\eta})$ is a positive operator, where $T$ again denotes the transpose map. Using the fact that $r_s\delta_t=\delta_{ts^{-1}}$ and $l_s\delta_t=\delta_{st}$ for all $s,t\in G$, we obtain
\begin{align*}0\leq(\id\ten T\circ\Theta(h))(x_{\eta,\eta})&=\frac{1}{|G|^2}\sum_{s,t,g\in G}(\id\ten T)(x_{\delta_s,\delta_t}\ten r_gx_{\delta_s,\delta_t}r_{g^{-1}})\\
                                                      &=\frac{1}{|G|^2}\sum_{s,t,g\in G}(\id\ten T)(x_{\delta_s,\delta_t}\ten l_sx_{\delta_{g^{-1}},\delta_{g^{-1}}}l_{t^{-1}})\\
                                                      &=\frac{1}{|G|^2}\sum_{s,t\in G}(\id\ten T)(x_{\delta_s,\delta_t}\ten l_{st^{-1}})\\
                                                      &=\frac{1}{|G|^2}\sum_{s,t\in G}x_{\delta_s,\delta_t}\ten l_{ts^{-1}}.\end{align*}
Now, observe that for any finite group $H$, an element $l(\vphi)=\sum_{s\in H}\vphi(s)l_s\in\mc{L}(H)$ is a positive operator if and only if $\vphi\in\mc{P}(H)$. Indeed, if $\xi\in\LT$, then
\begin{equation*}\la l(\vphi)\xi,\xi\ra=\sum_{s,t\in H}\vphi(s)\xi(s^{-1}t)\overline{\xi(t)}=\sum_{s,t\in H}\vphi(ts^{-1})\xi(s)\overline{\xi(t)}.\end{equation*}
Hence, positivity of $(\id\ten T\circ\Theta(h))(x_{\eta,\eta})$ ensures that $\tr((\id\ten T\circ\Theta(h))(x_{\eta,\eta})l(\vphi))\geq0$ for every $\vphi\in\mc{P}(G\times G)$. Therefore,
\begin{equation*}0\leq\frac{1}{\abs{G}^2}\sum_{s,t,g,h\in G}\vphi(g,h)\la l_g\delta_s,\delta_t\ra\tr(l_{ts^{-1}h})=\frac{1}{|G|}\sum_{g,h\in G}\vphi(g,h)\tr(l_{gh})=\sum_{g\in G}\vphi(g,g^{-1}).\end{equation*}
Taking $\vphi(g,h)=\la l_{(g,h)}\xi,\xi\ra$, $g,h\in G$, where $\xi\in L^2(G\times G)$, we see that $\sum_{g}l_{(g,g^{-1})}$ must be a positive operator. But as $\sum_{g}l_{(g,g^{-1})}=l(\chi)$ where $\chi$ is the characteristic function of $G\times G^{-1}:=\{(g,g^{-1})\mid g\in G\}$, this implies that $\chi$ is positive definite of norm 1. In particular, $G_{\chi}=\{g\in G\times G\mid\chi(g)=1\}=G\times G^{-1}$ is a subgroup of $G\times G$, which is the case if and only if $G$ is abelian.

It remains to show that $\mu$ is indeed equal to $h$. As $G$ is abelian, we may apply the Fourier transform to obtain the dual channel $\h{\Theta}(\h{\mu})$. As a consequence of \cite[Theorem 4]{HSR}, unitary equivalence preserves entanglement breaking, hence $\h{\Theta}(\h{\mu})$ is an entanglement breaking channel, and by Theorem \ref{EB}, $\h{\mu}=\delta_e=\h{h}$. Thus, by injectivity of the Fourier transform, $\mu=h$.\end{proof}

\begin{remark} As there exists a notion of entanglement breaking in infinite dimensions \cite{HSW}, it would be interesting to study the extent to which these two theorems generalize.\end{remark}

Let $G$ be a finite group with Haar measure $h$. Then $\h{\Theta}(\delta_e)$ and $\Theta(h)$ are conditional expectations of $\BLT$ onto $\LI$ and $\LG$, respectively. By composition, we can therefore restrict our quantum channels to the aforementioned subalgebras to obtain new examples of channels with positive MOE.\newpage

\begin{prop} Let $G$ be a finite group, and let $\mu\in M_1(G)$. Then $S_{\min}(\Theta(\mu)\circ\h{\Theta}(\delta_e))=H(\mu)$, the Shannon entropy of $\mu$.\end{prop}

\begin{proof} As $\Theta(\mu)\circ\h{\Theta}(\delta_e)$ is the restriction of $\Theta(\mu)$ to $\LI$, we can restrict the minimization of the von Neumann entropy to positive elements in $\LI$ of trace 1. Since these are of the form $\sum_{s\in G}\nu(s)x_{\delta_s,\delta_s}=\sum_{s\in G}\nu(s)l_sx_{\delta_e,\delta_e}l_{s^{-1}}$ for some $\nu\in M_1(G)$, and since $\Theta(\mu)$ is an $\LG$-bimodule map, we have
\begin{equation*}S\bigg(\Theta(\mu)\bigg(\sum_{s\in G}\nu(s)l_sx_{\delta_e,\delta_e}l_{s^{-1}}\bigg)\bigg)=S\bigg(\sum_{s\in G}\nu(s)l_s\Theta(\mu)(x_{\delta_e,\delta_e})l_{s^{-1}}\bigg)\geq S(\Theta(\mu)(x_{\delta_e,\delta_e}))\end{equation*}
by concavity of the von Neumann entropy. Now,
\begin{equation*}\Theta(\mu)(x_{\delta_e,\delta_e})=\sum_{s\in G}\mu(s)r_sx_{\delta_e,\delta_e}r_{s^{-1}}=\sum_{s\in G}\mu(s)x_{\delta_{s^{-1}},\delta_{s^{-1}}}=M_{\check{\mu}},\end{equation*}
where $\check{\mu}(s)=\mu(s^{-1})$ for $s\in G$ and we view $\mu$ as an element in $\LI$. Thus, $S_{\min}(\Theta(\mu)\circ\h{\Theta}(\delta_e))\geq S(M_{\check{\mu}})=H(\mu)$. The reverse inequality is obvious.\end{proof}

\begin{prop} Let $G$ be a finite group with Haar measure $h$, and let $\vphi\in\PO$. Then $S_{\min}(\h{\Theta}(\vphi)\circ\Theta(h))=S(\frac{1}{|G|}C_{\vphi})$, the entropy of the normalized correlation matrix of $\vphi$.\end{prop}

\begin{proof} As $\h{\Theta}(\vphi)\circ\Theta(h)$ is the restriction of $\h{\Theta}(\vphi)$ to $\LG$, we can restrict the minimization of the von Neumann entropy to positive elements in $\LG$ of trace 1. Since $\tr(l_s)=\delta_{s,e}|G|$, every state in $\LG$ is given by $\frac{1}{|G|}\sum_{s\in G}\sigma(s)l_s$ for some $\sigma\in\PO$ (see proof of Theorem \ref{EB2}). Letting $l_{\chi}=\frac{1}{|G|}\sum_{s\in G}l_s$, it follows that
\begin{equation*}\h{\Theta}(\vphi)(l_\chi)=\frac{1}{|G|}\sum_{s\in G}\vphi(s)l_s=\frac{1}{|G|}C_{\vphi},\end{equation*}
implying $S_{\min}(\h{\Theta}(\vphi))\leq S(\frac{1}{|G|}C_{\vphi})$. Invoking the algebraic properties of $\h{\Theta}$ and the fact that bistochastic channels increase entropy, we obtain
\begin{equation*}S(\h{\Theta}(\vphi)(\h{\Theta}(\sigma)(l_{\chi}))=S(\h{\Theta}(\sigma)(\h{\Theta}(\vphi)(l_{\chi})))\geq S(\h{\Theta}(\vphi)(l_{\chi})),\end{equation*}
implying the reverse inequality.\end{proof}

Given a finite-dimensional Hilbert space $H$ and a unital $^*$-subalgebra $M\subseteq\BH$, there exists a unique trace preserving conditional expectation $E:\BH\rightarrow M$ \cite[\S XI.Theorem 4.2]{T3}. As the previous results suggest, restricting quantum channels $\Phi:\BH\rightarrow\BH$ to $M$ via composition with $E$ not only yields new examples of channels, it also simplifies entropy calculations. Moreover, if $E$ is not entanglement breaking, as is the case for $\Theta(h)$ for nonabelian $G$ (Theorem \ref{EB2}), the composed channels $\Phi\circ E$ may provide a new approach to finding explicit counter-examples to the MOE conjecture.

\section{Outlook}

Shortly after the results of \cite{NRS} concerning $M(G)$ and $\Mcb$, Junge, Neufang and Ruan unified and generalized the representation theorems to the setting of \e{locally compact quantum groups} \cite{JNR}. As mentioned earlier, interesting work has already been done concerning the resulting quantum channels: in \cite{KNR}, their fixed points were completely characterized; in an early preprint of \cite{JNR}, arising from finite-dimensional quantum groups were analyzed, and their completely bounded minimal entropy was calculated; in the first author's Master's thesis \cite{C}, a method to generate canonical Kraus operators was obtained using the theory of unitary co-representations \cite[\S4.3]{C}, and an explicit description of the channels along with their fixed points was given for the smallest nontrivial quantum group, the (8-dimensional) Kac-Paljutkin algebra \cite[\S5.3]{C}.

Nevertheless, many interesting questions remain concerning these quantum group channels, such as: classical and quantum capacities, degradability, additivity, and applications to quantum error correction and quantum cryptography. There is also the practical consideration of whether these channels can be realized in the laboratory. This is certainly true for finite abelian groups, and, moreover, the construction of the Kac-Paljutkin channels suggests that they may be viewed as ``twisted'' Schur maps, so one is tempted to believe that such channels could arise naturally as perturbations of dephasing channels -- an exciting prospect for quantum group theory.

\end{spacing}

\end{document}